\documentclass[11pt]{article}

\usepackage[margin=1in]{geometry}
\usepackage{amsmath,amssymb,amsthm}
\usepackage{mathtools}
\usepackage{booktabs}
\usepackage{graphicx}
\usepackage{hyperref}
\usepackage{xcolor}
\usepackage{enumitem}
\usepackage{natbib}

\hypersetup{colorlinks=true,linkcolor=blue,citecolor=blue,urlcolor=blue}

\theoremstyle{plain}
\newtheorem{theorem}{Theorem}
\newtheorem{proposition}{Proposition}

\newtheorem{corollary}{Corollary}
\theoremstyle{definition}
\newtheorem{definition}{Definition}

\newtheorem{assumption}{Assumption}

\newcommand{\Var}{\operatorname{Var}}
\newcommand{\Cov}{\operatorname{Cov}}

\newcommand{\E}{\mathbb{E}}

\newcommand{\nbar}{\bar{n}}
\newcommand{\ybar}{\bar{y}}
\newcommand{\gbar}{\bar{g}}
\newcommand{\ebar}{\bar{e}}
\newcommand{\tauhat}{\hat{\tau}}
\newcommand{\Smacro}{S_{\mathrm{macro}}}
\newcommand{\Sres}{S_{\mathrm{res}}}
\newcommand{\Scl}{S_{\mathrm{cl}}}
\newcommand{\Stime}{S_{\mathrm{time}}}
\newcommand{\Sint}{S_{\mathrm{int}}}
\newcommand{\sigtot}{\sigma^2_{\mathrm{total}}}
\newcommand{\cv}{\mathrm{cv}}
\newcommand{\MSE}{\mathrm{MSE}}
\newcommand{\MSEw}{\MSE_{\mathrm{within}}}
\newcommand{\MSEm}{\MSE_{\mathrm{macro}}}
\newcommand{\Lpow}{\mathcal{L}_{\mathrm{power}}}
\DeclareMathOperator*{\argmin}{arg\,min}

\title{Power-Optimal Covariate Adjustment for Switchback Experiments}

\author{Sergei Pankratev \\ DoorDash, Inc.}

\date{}

\begin{document}
\maketitle

\begin{abstract}
In switchback experiments with unequal cluster sizes, outcome dispersion across randomization units inflates estimator variance and limits statistical power.
Standard control-using-prediction-as-covariate (CUPAC) adjustment may be suboptimal for variance reduction in this setting, because in its basic form it targets overall predictive accuracy and does not distinguish between the components of variance that vary across the randomization units of switchback experiments and those that vary across individual observations, even though these components contribute unequally to estimator variance.
We propose a power-optimal variance reduction methodology via CUPAC that balances prediction of the noise between and within randomization units to achieve maximum statistical power.
The methodology utilizes the framework for decomposition of the variance of the treatment-effect estimator for switchback experiments, and adapts both the outcome prediction and the analysis-time residualization to minimize the treatment-effect variance.
The study first develops the theoretical framework for the power-optimal CUPAC.
We then validate the theoretical framework through an extensive Monte Carlo simulation study.
Finally, we discuss the practical considerations of the proposed methodology, including its potential efficiency gains and limitations.
\end{abstract}

\section{Introduction}
\label{sec:intro}

Variance reduction via regression adjustment is one of the most effective techniques in online experimentation.
In standard settings, controlled-experiment using pre-experiment data (CUPED)~\citep{deng2013} adjusts the outcome using a pre-experiment covariate, while control using prediction as covariate (CUPAC) replaces the hand-chosen covariate with the prediction of a machine-learned model~\citep{poyarkov2016boosted}, so that the control variate is the conditional expectation of the outcome given a rich feature set~\citep{guo2021,jin2023}.
In both cases, the variance reduction realized at analysis time is governed by the correlation between the control variate and the outcome at the unit of analysis.
The control-variate model is, in practice, trained to minimize unit-level mean squared error (MSE).
This is a natural default in standard A/B tests: when units are randomized and analyzed at the same level, with no nested structure, there is a single relevant variance component, so the level at which the covariate predicts is immaterial and minimizing unit-level squared error directly maximizes the achievable variance reduction.

In switchback experiments, however, this standard default breaks down.
The randomization unit is a spatial cluster crossed with a time window---which we call a \emph{cell}---and these cells vary widely in the number of events they contain. It is this imbalance in cell size that amplifies how strongly the macro (between-cell) variation contributes to the variance of the unit-level metric.
Consequently, the variance of the treatment-effect estimator is dominated not by within-cell idiosyncratic noise but by the between-cell variation that the switchback design cannot average away~\citep{bojinov2023,pankratev2026powerful}.
A model trained on unit-level MSE may spend most of its capacity reducing within-cell error---where the squared-error mass lives---while doing comparatively little for the between-cell signal that actually drives statistical power. 
As a result, it can be suboptimal for variance reduction and power in this setting.

In this paper, we propose a unified framework for CUPAC training that maximizes statistical power in switchback experiments.
The framework re-weights both the covariate training objective and the analysis-time adjustment against the same measure of switchback variance.
Standard practice trains the control variate on unit-level MSE and estimates the adjustment coefficient by unit-level OLS; each step is optimal when randomization and analysis coincide at the unit level, but both are misaligned once treatment is assigned at the cell level and the metric aggregates finer-grained events.
Our framework corrects both stages: from the switchback design we derive the power-optimal training loss and the matched analysis coefficient, each re-weighted by exactly the factor the design dictates.

At its core is an exact decomposition of the switchback estimator's variance into a within-cell part and a between-cell part.
Because the design averages away within-cell noise inside each cell but cannot average away between-cell variation, it is the between-cell part that governs power---and its weight relative to the within-cell part grows with the mean cell size and with the imbalance in cell sizes.
This single observation identifies both \emph{where} standard practice goes wrong and \emph{how much} each stage must be re-weighted, and it is the common thread linking training and analysis.

As noted above, the framework has two components, one realigning each stage. The first component realigns covariate training.
Ordinary MSE implicitly treats within- and between-cell prediction errors as equally important, so a model fit to it spends most of its capacity chasing within-cell noise.
We instead train the control variate with a composite loss that deliberately up-weights between-cell error by the design-determined factor, steering the model's limited capacity toward the between-cell signal that actually drives power; this objective is a tight upper bound on the variance realized after analysis and applies directly to the model classes used in practice.

The second component realigns the analysis.
Since~\citet{deng2013}, the adjustment coefficient has been estimated by ordinary regression, which under-weights the between-cell relationship in exactly the same way, so pairing even a well-trained covariate with it leaves variance reduction on the table.
We show the coefficient should instead be estimated with the same between-cell emphasis used to train the covariate.
This closes the loop: the training builds between-cell predictive quality, and the matched analysis is what converts that quality into realized power.

We validate these predictions in a Monte Carlo study.
Across designs that range from balanced to heavily within-cell-dominated, the aligned framework recovers substantial power precisely where the theory predicts it should---as within-cell noise comes to dominate---while adding little where standard CUPAC already suffices.

Finally, we characterize when alignment yields the largest gains and discuss the main practical considerations it raises---chiefly the bias of the between-cell target in small cells and the overfitting risk that comes from concentrating the objective on a limited number of cells.

The paper is organized as follows.
Section~\ref{sec:theory} develops the theoretical framework: it formalizes the design and data-generating process, decomposes the prediction error, derives the estimator variance as an affine functional of the within- and between-cell losses, and---on that basis---aligns the two coordinated choices to the design, namely the power-optimal training loss and the matched analysis-time coefficient.
Section~\ref{sec:impl} provides the implementation details and practical considerations.
Section~\ref{sec:sims} reports a Monte Carlo simulation study validating the predicted alignment gains.
Section~\ref{sec:discussion} discusses implications and limitations, and Section~\ref{sec:conclusion} concludes.

\section{Literature Review}
\label{sec:related}

This work sits at the intersection of two literatures: variance reduction for online controlled experiments, and switchback and cluster-randomized designs.

Variance reduction through regression adjustment is a mature line of work in online experimentation.
CUPED introduced the use of pre-experiment covariates to reduce the variance of the treatment-effect estimator without biasing it~\citep{deng2013}.
CUPAC generalized the fixed covariate to a machine-learned prediction of the outcome, using boosted trees to raise the covariate's correlation with the outcome and hence the achievable variance reduction~\citep{poyarkov2016boosted,li2024cupac}.
Subsequent work has studied machine-learned control variates more broadly and characterized the variance-optimal use of a given predictor~\citep{guo2021,jin2023}.
The underlying principle---calibrating a predictor so as to minimize the variance of a downstream estimator rather than its own prediction error---is shared with prediction-powered inference~\citep{angelopoulos2023}.
These methods, and empirical comparisons among them, have been developed almost exclusively in the standard A/B setting with independent user-level randomization~\citep{staponaite2025variance,kohavi2020trustworthy}, where the unit of randomization and the unit of analysis coincide and unit-level prediction error is the appropriate training target.
A companion study evaluates these methods---CUPED, CUPAC, and doubly robust estimation---directly in the switchback setting, mapping across design regimes when each yields the largest efficiency gains~\citep{pankratev2026design}; it quantifies the variance reduction available but treats the covariate model as given rather than asking how it should be trained.

A second literature studies the design and analysis of switchback experiments, in which treatment is assigned to cluster$\times$time-window cells~\citep{bojinov2019time,bojinov2023}.
This work has focused primarily on point estimation, randomization inference, and the choice of switching interval under temporal interference~\citep{xiong2023data,hu2022switchback}, rather than on covariate adjustment.
A recurring theme is that power in clustered designs is governed not by the number of observations but by the number of randomization units, and is further eroded by imbalance in cluster sizes.
This cluster-size-imbalance penalty is the design-effect factor of classical survey sampling~\citep{kish1965survey} and cluster-randomized trials~\citep{eldridge2006sample}, in which unequal cluster sizes inflate the variance of the treatment-effect estimator by a factor that grows quadratically with the dispersion in cluster sizes.
\citet{pankratev2026powerful} makes this explicit for switchbacks, deriving the closed-form variance decomposition---separating within-cell from macro (between-cell) variance and carrying a penalty that grows with the imbalance in cluster sizes---on which the present analysis builds.

These two literatures have developed largely independently.
Variance-reduction methods are trained, evaluated, and applied at the unit level, which is optimal when randomization is individual; switchback power analyses characterize the variance of the unadjusted estimator but do not ask how a control variate should be trained to minimize it.
The present work connects them by recasting the switchback variance functional as a prediction loss and re-weighting the CUPAC workflow---covariate training and analysis-time adjustment---against the same design quantity $1+\nbar(1+\cv^2)$, an alignment that, to our knowledge, prior CUPAC practice does not pursue.

\section{Theoretical Framework}
\label{sec:theory}

The methodology is motivated by a single structural feature of switchback experiments: the variance of the treatment-effect estimator is dominated by macro (between-cell) variation rather than by within-cell idiosyncratic noise, and the contribution of that macro variation is further amplified by cluster-size imbalance through the factor $1+\cv^2$.
A control variate trained to minimize unit-level error therefore spends most of its capacity on the within-cell noise that the design averages away, while doing comparatively little for the between-cell signal that governs power.
Correcting this misalignment is not the task of the training loss alone: realigning the loss is necessary but only partial, because the analysis-time residualization coefficient is subject to the same distortion, and it contributes an additional layer of achievable variance reduction.
This section builds the framework that makes these statements precise and develops the aligned corrections.
We first formalize the design and data-generating process, then decompose unit-level prediction error into orthogonal within-cell and between-cell losses and express the estimator variance as an affine functional of the two.
This identity exposes a single design quantity, $1+\nbar(1+\cv^2)$, by which the conventional objectives under-weight the between-cell component, and it is the common thread through the two coordinated choices we then align to it: the training loss and the analysis-time coefficient.

\subsection{Setup and Notation}
\label{sec:setup}

We model the outcome of interest via a multilevel data-generating process (DGP).
Randomization cells are indexed by $b=(\mathrm{cl},t)$, a spatial cluster $\mathrm{cl}$ crossed with a time window $t$, and the $n_b$ units inside cell $b$ are indexed by $i=1,\dots,n_b$.
A generic observation is thus the pair $(b,i)$, and every unit-level quantity carries the cell $b$ in its subscript so that membership is explicit; cell-level aggregates carry the subscript $b$ alone.
The outcome of unit $i$ in cell $b$ is modeled as:
\begin{equation}
\label{eq:model}
y_{b,i} = \mu + \alpha_{\mathrm{cl}} + \beta_t + \gamma_{\mathrm{cl},t} + \varepsilon_{b,i}.
\end{equation}
Here $\mu$ is the grand mean, $\alpha_{\mathrm{cl}}$ is the spatial cluster effect, $\beta_t$ is the time-window effect, $\gamma_{\mathrm{cl},t}$ is the cluster-by-time interaction, and $\varepsilon_{b,i}$ is the individual idiosyncratic deviation.
The $\alpha,\beta,\gamma$ terms are random and differ across randomization cells but are constant within them, whereas $\varepsilon_{b,i}$ varies across individual units within a cell.

The four components are mutually orthogonal, with variances $\Var(\alpha)=\sigma^2_{\mathrm{cl}}$, $\Var(\beta)=\sigma^2_{\mathrm{time}}$, $\Var(\gamma)=\sigma^2_{\mathrm{int}}$, and $\Var(\varepsilon)=\sigma^2_{\mathrm{res}}$, so the total outcome variance is $\sigtot = \sigma^2_{\mathrm{cl}}+\sigma^2_{\mathrm{time}}+\sigma^2_{\mathrm{int}}+\sigma^2_{\mathrm{res}}$.
The idiosyncratic within-cell variance $\sigma^2_{\mathrm{res}}$ represents the residual variance, capturing the variation of individual units within any given cell.
Conversely, the sum of the spatial, temporal, and interaction variances, $\sigma^2_{\mathrm{macro}} = \sigma^2_{\mathrm{cl}} + \sigma^2_{\mathrm{time}} + \sigma^2_{\mathrm{int}}$, represents the macro variance, which aggregates all variation operating at the cell level.
We define the variance \emph{shares} $S_k = \sigma^2_k/\sigtot$ and group the cell-level (``macro'') share:
\[
\Smacro = \Scl + \Stime + \Sint,
\qquad
\Sres = 1-\Smacro .
\]

The experiment is run over $J$ spatial clusters and $H$ time windows (hours), yielding $B=JH$ randomization \emph{cells}.
Treatment is assigned at the cell level: $W_b\in\{0,1\}$, balanced across the design.
Because cell sizes are typically non-uniform in practice, each cell $b$ contains $n_b$ measurement units.
We characterize the distribution of cell sizes by the mean cell density $\nbar$ and the squared coefficient of variation (CV) of cell sizes:
\[
\nbar = \frac{1}{B}\sum_{b=1}^B n_b,
\qquad
\cv^2 = \frac{1}{\nbar^2}\cdot\frac{1}{B}\sum_{b=1}^B (n_b-\nbar)^2 .
\]

The raw estimator variance follows directly from this design and DGP.
For the individual-level ordinary least squares (OLS) difference-in-means estimator $\tauhat$, the switchback power formula~\citep{pankratev2026powerful} gives:
\begin{equation}
\label{eq:rawvar}
\Var(\tauhat)\;\approx\;\frac{4\,\sigtot}{JH}
\left[\frac{\Sres}{\nbar}
+ \Smacro\!\left(\frac{1}{\nbar}+1+\cv^2\right)\right].
\end{equation}
The structural asymmetry between these two bracketed terms is central to our development: the residual (within-cell) variance is divided by $\nbar$, while the macro share is multiplied by $\tfrac{1}{\nbar}+1+\cv^2$.
For $\nbar$ in the hundreds, the macro share is weighted roughly $\nbar(1+\cv^2)$ times more heavily.
The $(1+\cv^2)$ factor is the cluster-size-imbalance penalty familiar from cluster-randomized trials, where unequal cluster sizes inflate the variance of the treatment-effect estimator~\citep{eldridge2006sample}.

CUPED and CUPAC introduce a control variate $g_{b,i}=g(X_{b,i})$---a pre-experiment covariate under CUPED, or a model prediction under CUPAC---where the pre-experiment features $X_{b,i}$ are correlates of either the residual within-cell variance or the variance across randomization cells.
The adjusted outcome is then formed as:
\begin{equation}
\label{eq:cv}
y_{b,i}^{\mathrm{adj}} = y_{b,i} - \theta\,(g_{b,i} - \E[g]),
\end{equation}
where $\theta$ is estimated at analysis time.
The estimator is then re-computed on $y^{\mathrm{adj}}$.
Our central object of study is the predictor $g$, and the guiding question: \emph{what loss should $g$ minimize so that $\Var(\tauhat)$ is minimized?}
Minimizing the same estimator variance $\Var(\tauhat)$ also shapes two companion choices that we study alongside $g$---the analysis-time coefficient $\theta$ and the weighting of observations by cell size---each contributing a further layer of variance reduction.

\subsection{Error Decomposition}

To formalize the predictive models at both the individual and cell levels, we define the respective outcomes and residuals.
At the individual level, the outcome $y_{b,i}$ for unit $i$ in cell $b$ is modeled as:
\begin{equation}
\label{eq:indiv_model}
y_{b,i} = g(X_{b,i}) + e_{b,i},
\end{equation}
where $g(X_{b,i})$ (or $g_{b,i}$) is the individual-level prediction and $e_{b,i}$ is the individual residual. To establish the connection to the underlying DGP~\eqref{eq:model}, let $M_b = \alpha_{\mathrm{cl}} + \beta_t + \gamma_{\mathrm{cl},t}$ represent the total cell-constant macro variation.
We can then equate the outcomes to write:
\begin{equation}
\label{eq:dgp_predictive}
g(X_{b,i}) + e_{b,i} = \mu + M_b + \varepsilon_{b,i},
\end{equation}
which implies that the prediction residual satisfies:
\begin{equation}
\label{eq:residual_prim}
e_{b,i} = (\mu + M_b + \varepsilon_{b,i}) - g(X_{b,i}) .
\end{equation}
Under this representation, the pre-experiment features $X_{b,i}$ may contain covariates that correlate with either the cell-constant macro component $M_b$ (the between-cell signal) or the individual idiosyncratic variation $\varepsilon_{b,i}$ (the within-cell noise).
If we were to aggregate the data and train a separate model directly at the cell level, the cell-mean outcome $\bar{y}_b = \frac{1}{n_b} \sum_{i \in b} y_{b,i}$ would be modeled as:
\begin{equation}
\label{eq:cell_model}
\bar{y}_b = h(X_b) + u_b,
\end{equation}
where $h(X_b)$ is the cell-level prediction and $u_b$ is the cell-level residual. In contrast, when we aggregate the individual predictions of $g(X_{b,i})$ post-hoc, the predicted cell mean is $\bar{g}_b = \frac{1}{n_b} \sum_{i \in b} g_{b,i}$, and the cell-mean residual is $\bar{e}_b = \bar{y}_b - \bar{g}_b$.

Decompose the individual residual into orthogonal within-cell and between-cell parts:
\begin{equation}
\label{eq:errdecomp}
e_{b,i} = \underbrace{(e_{b,i}-\ebar_b)}_{\text{within-cell}} + \underbrace{\ebar_b}_{\text{macro}}.
\end{equation}
By substituting~\eqref{eq:residual_prim} into~\eqref{eq:errdecomp}, we can express both residual components directly in terms of the DGP primitives.
For the within-cell residual, because $y_{b,i} - \ybar_b = \varepsilon_{b,i} - \bar{\varepsilon}_b$, we have:
\begin{equation}
\label{eq:within_dgp}
e_{b,i} - \ebar_b = (\varepsilon_{b,i} - \bar{\varepsilon}_b) - (g_{b,i} - \gbar_b) ,
\end{equation}
which depends purely on how well the within-cell variation of the predictor explains individual-level idiosyncratic noise.
For the macro residual, the cell-mean prediction residual satisfies:
\begin{equation}
\label{eq:macro_dgp}
\ebar_b = \big( \mu + M_b - \gbar_b \big) + \bar{\varepsilon}_b ,
\end{equation}
which represents the unpredicted macro-level variation $\mu + M_b - \gbar_b$ contaminated by the cell-average within-cell sampling noise $\bar{\varepsilon}_b$.
This formulation shows that the two prediction losses defined in~\eqref{eq:losses} below are driven by distinct structural components of the DGP.
We define these two prediction losses:
\begin{equation}
\label{eq:losses}
\MSEw(g) = \E\big[(e_{b,i}-\ebar_b)^2\big]
        = \E\big[\big((y_{b,i}-\ybar_b)-(g_{b,i}-\gbar_b)\big)^2\big],
\qquad
\MSEm(g) = \E\big[\,\ebar_b^{\,2}\,\big]
        = \E\big[(\ybar_b-\gbar_b)^2\big].
\end{equation}
$\MSEw$ is the within-cell mean squared error (how well $g$ tracks within-cell deviations); $\MSEm$ is the between-cell mean squared error (how well the cell-mean prediction $\gbar_b$ tracks the cell-mean outcome $\ybar_b$).
Because the decomposition~\eqref{eq:errdecomp} is orthogonal,
\begin{equation}
\label{eq:totaldecomp}
\MSE_{\mathrm{total}}(g) := \E[(y_{b,i}-g_{b,i})^2] = \MSEw(g) + \MSEm(g).
\end{equation}

\subsection{Power-Aligned Training Loss}
\label{sec:loss}

With the within- and between-cell losses in hand, we can express what actually governs power.
The result below shows that the adjusted estimator variance is a fixed linear combination of $\MSEw$ and $\MSEm$, with weights determined entirely by the design---so reducing variance amounts to reducing a specific weighted sum of the two losses.
To isolate the effect of $g$ on the estimator variance, we first fix the adjustment coefficient at the calibrated value $\theta=1$; this simplifying choice is relaxed, and the coefficient optimized, in Section~\ref{sec:theta}.

\begin{assumption}[Calibrated control variate]
\label{ass:calib}
The control variate $g$ of~\eqref{eq:cv} is used with coefficient $\theta=1$ and is conditionally unbiased, $\E[y_{b,i}\mid X_{b,i}]=g_{b,i}$.\footnote{Under these conditions the adjusted outcome equals the residual up to an additive constant: from~\eqref{eq:cv} with $\theta=1$, $y_{b,i}^{\mathrm{adj}}=y_{b,i}-(g_{b,i}-\E[g])=(y_{b,i}-g_{b,i})+\E[g]=e_{b,i}+\E[g]$, and the constant $\E[g]$ is common to both treatment arms, so it cancels in the difference-in-means and does not affect $\Var(\tauhat)$.}
\end{assumption}

\begin{proposition}[Power loss identity]
\label{thm:variance}
Under the multilevel model~\eqref{eq:model} and Assumption~\ref{ass:calib}, the variance of the CUPAC-adjusted estimator satisfies
\begin{equation}
\label{eq:cupacvar}
\Var(\tauhat_{\mathrm{CUPAC}})
\;\approx\;
\frac{4}{JH}\left[\frac{\MSEw(g)}{\nbar}
+ \MSEm(g)\!\left(\frac{1}{\nbar}+1+\cv^2\right)\right].
\end{equation}
Equivalently, $\Var(\tauhat_{\mathrm{CUPAC}}) \approx \tfrac{4}{JH}\,\Lpow(g)$ with $\Lpow$ defined in~\eqref{eq:Lpow} below.
\end{proposition}

The identity is immediate: under calibration the adjusted outcome is just the residual $e_{b,i}=y_{b,i}-g_{b,i}$, so substituting its within- and between-cell variance components (which are $\MSEw$ and $\MSEm$) for the a-priori shares $\Sres,\Smacro$ in the raw power formula~\eqref{eq:rawvar} gives~\eqref{eq:cupacvar} (Appendix~\ref{app:proofs}).

The identity replaces the \emph{a-priori} shares $\Sres,\Smacro$ of the raw formula by the \emph{post-adjustment residual} losses $\MSEw,\MSEm$: the control variate reduces estimator variance precisely by shrinking these two residual losses, each weighted by its design multiplier.
We write these multipliers as
\begin{equation}
\label{eq:abweights}
a = \frac{1}{\nbar}
\qquad\text{(within-cell weight)},
\qquad
b = \frac{1}{\nbar}+1+\cv^2
\qquad\text{(between-cell weight)},
\end{equation}
so that $\Var(\tauhat_{\mathrm{CUPAC}})\approx\tfrac{4}{JH}\big[a\,\MSEw(g)+b\,\MSEm(g)\big]$, and refer to $a$ and $b$ throughout the remainder of the paper.

The design weights $a$ and $b$ of~\eqref{eq:abweights} are far from equal, yet the objective used to train $g$ in practice is blind to this asymmetry.
Plain unit-level MSE splits as $\MSEw+\MSEm$ and so implicitly assigns the within- and between-cell losses equal weight.

\begin{corollary}[Implicit weighting of the conventional objective]
\label{prop:misalign}
The unweighted unit-level loss $\MSE_{\mathrm{total}}=\MSEw+\MSEm$ assigns the within- and between-cell losses equal effective weight ($b/a=1$), whereas Proposition~\ref{thm:variance} requires
\begin{equation}
\label{eq:ratio}
\frac{b}{a}=\frac{\tfrac{1}{\nbar}+1+\cv^2}{\tfrac{1}{\nbar}}=1+\nbar(1+\cv^2);
\end{equation}
it therefore under-weights $\MSEm$ by the factor $1+\nbar(1+\cv^2)$ (proof in Appendix~\ref{app:proofs}).
\end{corollary}

The conventional objective thus never controls \emph{how much} it penalizes the between-cell error; supplying that correction is the role of the power-optimal loss.

The power-optimal loss adopts these design weights directly.

\begin{definition}[Power-optimal loss]
\label{def:Lpow}
Given design constants $\nbar$ and $\cv^2$, define
\begin{equation}
\label{eq:Lpow}
\Lpow(g)
= \frac{1}{\nbar}\,\MSEw(g)
+ \left(\frac{1}{\nbar}+1+\cv^2\right)\MSEm(g).
\end{equation}
\end{definition}

Because $\Var(\tauhat_{\mathrm{CUPAC}})\approx\tfrac{4}{JH}\Lpow(g)$ is an equality up to the fixed constant $4/JH$ under Assumption~\ref{ass:calib}, minimizing $\Lpow$ \emph{is} minimizing the calibrated estimator variance; once $\theta$ is estimated rather than fixed at $1$, the relevant guarantee is the upper bound of Theorem~\ref{thm:upperbound}.
Rescaling by $\nbar$ gives an interpretable \emph{penalty form}, $\Lpow(g)\propto\MSE_{\mathrm{total}}(g)+\lambda\,\MSEm(g)$ with $\lambda=\nbar(1+\cv^2)$: train on ordinary MSE, then add a between-cell penalty whose strength is the mean cell density inflated by cell-size dispersion---a measurable design quantity, not a tuning knob.
The weights depend only on the two design constants $\nbar$ and $\cv^2$, which are read directly from the experiment's cell-size distribution; they are metric- and segment-specific, since subset metrics (e.g.\ regulatory vs.\ non-regulatory markets) have different cell densities and hence different $\lambda$, so a control variate trained per target metric uses its own $(\nbar,\cv^2)$.

\subsection{Power-Aligned Analysis-Time Covariate Coefficient}
\label{sec:theta}

The analysis so far optimizes the predictor $g$ while holding the adjustment coefficient at the calibrated value $\theta=1$ (Assumption~\ref{ass:calib}).
In deployment, however, $\theta$ is not fixed: it is \emph{estimated} from the experiment data, almost always by least squares.
Optimizing the covariate while estimating the coefficient by conventional unit-level regression re-introduces the very gap $\Lpow$ was designed to close, and it can make a power-optimally trained covariate appear to yield no advantage.
The reason, as we now show, is that the choice of estimator for $\theta$ is subject to \emph{exactly the same} switchback misalignment as the choice of training loss for $g$.

Adjusting the outcome with a coefficient $\theta$ yields the residual $r_{b,i}(\theta)=y_{b,i}-\theta g_{b,i}$ (at $\theta=1$, the residual $e_{b,i}$ of Assumption~\ref{ass:calib}), which splits into within- and between-cell parts exactly as in the error decomposition~\eqref{eq:errdecomp}.
Applying Proposition~\ref{thm:variance} to $r(\theta)$, the estimator variance $\tfrac{JH}{4}\Var(\tauhat_{\mathrm{CUPAC}}(\theta))=a\,\MSEw(r(\theta))+b\,\MSEm(r(\theta))$ is a strictly convex quadratic in $\theta$, minimized at
\begin{equation}
\label{eq:thetastar}
\theta^\star = \frac{a\,C^w + b\,C^m}{a\,V^w_g + b\,V^m_g},
\end{equation}
where $C^w,C^m$ are the within- and between-cell covariances between outcome and prediction and $V^w_g,V^m_g$ the corresponding within- and between-cell variances of $g$ (the second-moment definitions and the derivation are collected in Appendix~\ref{app:proofs}).
Equivalently, $\theta^\star$ is the slope from regressing $y$ on $g$ under the composite loss, $\theta^\star=\argmin_\theta \Lpow(y-\theta g)$: estimating $\theta$ with the \emph{same} loss used to train $g$ yields the variance-minimizing coefficient.

The standard coefficient inherits the same misalignment.
The textbook CUPED/CUPAC coefficient is the unit-level OLS slope, which by the law of total (co)variance is
\begin{equation}
\label{eq:thetaunit}
\theta_{\mathrm{unit}} = \frac{\Cov(y_{b,i},g_{b,i})}{\Var(g_{b,i})} = \frac{C^w+C^m}{V^w_g+V^m_g}.
\end{equation}
Comparing~\eqref{eq:thetaunit} with~\eqref{eq:thetastar}, $\theta_{\mathrm{unit}}$ is exactly the special case $a=b$: it weights the within- and between-cell covariances equally, whereas the variance-optimal coefficient weights the between-cell covariance more heavily by the factor $b/a$.

\begin{proposition}[Analysis-stage misalignment]
\label{prop:theta_misalign}
Relative to the variance-optimal coefficient $\theta^\star$, unit-level OLS estimation of $\theta$ under-weights the between-cell covariance $C^m$ by exactly the factor
\[
\frac{b}{a} = 1+\nbar(1+\cv^2),
\]
the identical misalignment that Corollary~\ref{prop:misalign} establishes for unit-level training of $g$.
\end{proposition}

The $1{:}1$ versus $1{:}\big(1+\nbar(1+\cv^2)\big)$ misalignment is thus not a property of the training objective in particular; it is a property of \emph{any} unit-level least-squares operation on switchback data, and it strikes twice---once when fitting $g$, and again when fitting $\theta$.
A macro-targeted covariate paired with a unit-level coefficient is throttled at the analysis stage, which explains why a power-optimally trained control variate can fail to realize its predicted variance reduction unless $\theta$ is re-estimated with a macro-aware loss.

Matched estimation yields the best \emph{single} coefficient, but a single coefficient is itself a constraint.
Allowing a separate coefficient per variance level---$\theta_w$ on within-cell deviations and $\theta_m$ on cell means---gives
\begin{equation}
\label{eq:perlevel}
\min_{\theta_w,\theta_m}\mathcal V
= a\,V^w_y\big(1-\rho_w^2\big) + b\,V^m_y\big(1-\rho_m^2\big),
\qquad
\rho_w=\frac{C^w}{\sqrt{V^w_y V^w_g}},\quad
\rho_m=\frac{C^m}{\sqrt{V^m_y V^m_g}},
\end{equation}
attained at $\theta_w^\star=C^w/V^w_g$ and $\theta_m^\star=C^m/V^m_g$.
This is the global optimum over all linear adjustments, and it depends on $g$ \emph{only through the per-level correlations} $\rho_w,\rho_m$: once $\theta$ is fit per level, the calibration \emph{scale} of $g$ is irrelevant.

\begin{corollary}[When a single coefficient suffices]
\label{cor:perlevel}
The single matched coefficient $\theta^\star$ attains the per-level optimum~\eqref{eq:perlevel} if and only if $g$ is \emph{level-calibrated}, i.e.\ the $y$-on-$g$ regression slope is the same within and between cells, $C^w/V^w_g = C^m/V^m_g$.
Otherwise a single coefficient strictly under-performs per-level adjustment.
\end{corollary}

Two consequences follow, and together they sharpen the message of the training-loss analysis above.
First, because the attainable variance~\eqref{eq:perlevel} depends on $g$ only through $\rho_w,\rho_m$, the training loss can improve switchback power \emph{only insofar as it raises the between-cell correlation} $\rho_m$.
Two objectives that differ merely in how they \emph{scale} the macro component of $g$---but induce the same $\rho_m$---are indistinguishable once $\theta$ is fit per level, because the coefficient absorbs the scale.
The power-optimal loss therefore helps not by rescaling $g$ (the coefficient undoes scaling) but by reallocating finite model capacity toward raising $\rho_m$ whenever capacity or regularization forces a trade-off between the within- and between-cell fits; absent such a trade-off, all reasonable losses reach the same $\rho_m$ and hence the same optimal power.
Second, the practically robust recipe is to align \emph{both} stages: train $g$ to raise $\rho_m$ under $\Lpow$, and estimate $\theta$ per variance level (equivalently, under the matched composite loss), so that the realized variance reduction is insensitive to any residual miscalibration of $g$.

We have so far treated two levers separately: the loss that trains the covariate $g$, and the coefficient $\theta$ that residualizes with it at analysis time.
These interact, so we must check that optimizing them separately is coherent.
On the one hand, $\Lpow(g)$ is what we minimize when \emph{training} $g$, and it is measured at the calibrated coefficient $\theta=1$ (Assumption~\ref{ass:calib}).
On the other hand, the variance actually realized in \emph{deployment} uses the best coefficient the second lever can supply---the per-level optimum $(\theta_w^\star,\theta_m^\star)$ of~\eqref{eq:perlevel}.
Write $\tauhat^{\,\theta^\star}_{\mathrm{CUPAC}}$ for that deployed estimator: the CUPAC estimator with $g$ held fixed and $\theta$ fit optimally per variance level.

The concern is that these two could disagree---that pushing down the training surrogate $\Lpow(g)$ at $\theta=1$ might not lower the variance realized once $\theta$ is re-fit.
The following theorem rules this out: the training surrogate is always an upper bound on the deployed variance, so minimizing $\Lpow$ can only help.

\begin{theorem}[Upper-bound surrogate]
\label{thm:upperbound}
Let $\rho_k$ denote the correlation between the level-$k$ component of $y$ and of $g$, for $k$ ranging over the within-cell and between-cell levels, and let $\Var(\tauhat^{\,\theta^\star}_{\mathrm{CUPAC}})$ be the estimator variance when $\theta$ is fit at the level-optimal values $\theta_k^\star$.
Then
\[
\frac{JH}{4}\,\Var(\tauhat^{\,\theta^\star}_{\mathrm{CUPAC}})
\;\le\; \Lpow(g),
\]
with equality iff $g$ is \emph{calibrated} at every level, i.e.\ the level-$k$ regression slope of $y$ on $g$ equals $1$ ($\sigma_{g,k}=\rho_k\sigma_{y,k}$).
\end{theorem}
\noindent\emph{Proof idea.} At each level the MSE component dominates the optimal-$\theta$ residual variance $\sigma^2_{y,k}(1-\rho_k^2)$, with equality when $g$ is calibrated at that level; summing across levels with the design weights of~\eqref{eq:cupacvar} gives the bound (Appendix~\ref{app:proofs}).

In words: minimizing $\Lpow$ (the first lever) minimizes an upper bound on the variance that survives re-fitting $\theta$ (the second lever), so the two levers do not work against each other.
Keeping $g$ calibrated (e.g.\ via isotonic or Platt recalibration) makes the bound tight, and re-fitting $\theta$ at analysis time recovers any remaining slack for free.

\section{Implementation and Practical Considerations}
\label{sec:impl}

The framework of Section~\ref{sec:theory} prescribes the three aligned corrections but is silent on when they are worth adopting and how to deploy them without introducing new problems.
This section takes up both questions.
We first discuss the covariate structures under which alignment actually moves power---which hinges on whether a within-versus-between capacity trade-off binds---and then turn to the practical hazards of estimating a macro-weighted loss from the comparatively small number of cells, together with the design choices that mitigate them.

\subsection{When Alignment Helps: Feature Structure and Model Capacity}
\label{sec:whenuseful}

The gains from power-aligned training are not uniform: they depend on the structure of the available covariates and on whether model capacity binds.
The per-level analysis of Equation~\eqref{eq:perlevel} makes the condition precise, since the attainable estimator variance depends on the covariate only through the within- and between-cell correlations $\rho_w$ and $\rho_m$.
A training loss can therefore improve power only insofar as it raises $\rho_m$, the correlation between the covariate's cell means and the cell-mean outcome, and whether $\Lpow$ does so over a unit-level loss hinges entirely on the feature set.

When the features that predict the cell-mean signal are abundant, cheap, and distinct from those that predict within-cell noise---so that a model can drive both $\rho_w$ and $\rho_m$ to their ceilings at once---the training objective is immaterial: all reasonable losses reach the same per-level correlations, and hence the same power once $\theta$ is fit per level.
The composite loss earns its keep only when a genuine within-versus-between capacity trade-off binds, that is, when macro-predictive features must compete for finite model capacity (or a fixed regularization budget) with a much larger set of within-cell-predictive features.
In that case a unit-level objective spends its capacity on the within-cell signal, where the squared-error mass lives, and sacrifices the very $\rho_m$ that governs power---exactly the regime the simulation study of Section~\ref{sec:sims} isolates using two covariates with deliberately conflicting within- and between-cell profiles under a fixed Ridge penalty.

This binding case is the common one in practice, because the feature set for a switchback covariate typically mixes a small number of cell-level (macro) predictors---such as lagged cell-mean rates, cluster attributes, or time-of-day effects---with a much larger set of fine-grained unit-level predictors.
The mechanism by which standard training then discards the macro predictors is not specific to any one model, but recurs across model classes.
In any model with finite capacity or regularization constraints, standard training objectives allocate capacity to reduce the largest source of loss.
Because individual-level residual variance $\sigma^2_{\mathrm{res}}$ dominates the unweighted sum of squares, and because the number of individual observations $N$ is orders of magnitude larger than the number of cells $B$, the objective function is highly insensitive to cell-level predictions.

First, in linear models (such as OLS, Ridge, or Lasso), any L1 or L2 regularization penalty will shrink the coefficients of cell-level features to zero long before it shrinks individual-level features, because the marginal reduction in unit-level MSE from cell-level features is too small to overcome the shrinkage penalty.
Second, in tree-based models (such as CART, Random Forests, or LightGBM), split criteria greedily maximize unit-level sum-of-squared-error (SSE) reduction.
An individual-level split that reduces within-cell variance by a tiny fraction yields a total SSE reduction that scales with $N$, whereas a cell-level split that explains a large fraction of macro variance yields a reduction that scales only with $B$.
Consequently, standard decision trees will greedily split on individual-level features first, completely ignoring cell-level features before hitting depth limits, while standard hyper-parameters (such as minimum child weight or L2 leaf regularization) will actively prune out cell-level splits.
Third, in neural networks (such as Multi-Layer Perceptrons), weight decay and early stopping will decay weights for cell-level features to zero because their gradients are dominated by the high-frequency individual-level patterns.
By up-weighting the cell-level prediction error by $\lambda = \nbar(1+\cv^2)$, the power-optimal loss $\Lpow$ scales up the gradients, Hessians, and SSE reductions of cell-level features, forcing any of these algorithms to prioritize and preserve them.
The practical implication is a simple diagnostic: the methodology is worth adopting precisely when the covariate set contains genuine cell-level predictors that a unit-level objective would otherwise leave on the table.

\subsection{Practical Considerations}
\label{sec:practice}

The effective sample size behind the macro loss governs its overfitting risk.
The within-cell loss is supported on $N$ (millions of) units, but the macro loss is supported on only $B=JH$ cells.
Up-weighting $\MSEm$ by $\lambda=\nbar(1+\cv^2)$ therefore concentrates the effective objective on a much smaller sample, raising overfitting risk for the between-cell component.
For this reason, model selection is more reliable when based on a \emph{cell-level temporal (or cluster) holdout} rather than a random unit split, since the apparent macro gains estimated on a random split tend not to transfer to the experiment window.

The macro target is itself biased in small cells.
The empirical cell mean satisfies $\Var(\ybar_b)=\sigma^2_{\mathrm{macro}}+\sigma^2_{\mathrm{res}}/n_b$, so for small cells the macro target $\ybar_b$ is contaminated by within-cell sampling noise.
This biases $\MSEm$ upward and is most severe for low-density segments (e.g.\ regulatory markets, where $\nbar$ is small).
Mitigations: (i) weight the macro term by $n_b$ (inverse-variance weighting of the cell mean), or (ii) shrink the target toward a hierarchical (empirical-Bayes) cell-mean estimate.
Note that small $\nbar$ also \emph{lowers} the optimal $\lambda$ via~\eqref{eq:ratio}, so the relative importance of within-cell prediction genuinely rises in low-density segments---the formula prescribes the correct trade-off.

\section{Simulation Study}
\label{sec:sims}

The theory makes two sharp, testable predictions: alignment should be immaterial when the macro signal is abundant, but should recover an increasing share of power as within-cell noise comes to dominate, and it should help only when a genuine within-versus-between capacity trade-off binds.
We test both in a controlled Monte Carlo study that reproduces the multilevel design of Section~\ref{sec:theory} and pits the naive estimator against the fully aligned one---and against each correction in isolation---across a sequence of macro-variance shares.
We first describe the data-generating process and the estimators under comparison, then report the results.%

\subsection{Simulation Design}
\label{sec:sim_design}

We simulate switchback experiments from the multilevel model of Equation~\eqref{eq:model}, drawing the cluster, time, interaction, and residual components as mutually independent Gaussian variates and the within-cell measurement counts from a lognormal--Poisson mixture.
The design comprises $J=200$ clusters and $H=24$ time windows, for $B=JH=4{,}800$ randomization cells, with a mean cell density of $\nbar=180$ measurement units and a cell-size coefficient of variation of $\cv=1.5$; treatment is assigned independently at the cell level with probability one half and enters the outcome additively.%

The central design axis is the macro variance share $\Smacro$.
We hold the total outcome variance fixed and consider three regimes, $\Smacro\in\{0.50,\,0.25,\,0.15\}$, distributing $\Smacro$ across the cluster, time, and interaction components in fixed proportion---so that $\Smacro=0.50$ reproduces $(\Sint,\Scl,\Stime)=(0.20,0.15,0.15)$---and assigning the remainder to the residual share $\Sres=1-\Smacro$.
Decreasing $\Smacro$ moves the design toward the idiosyncratic-noise-dominated regime that is characteristic of fine-grained event data, in which the within-cell residual accounts for the large majority of unit-level variance.%

The control variate is a Ridge regression on two engineered covariates with deliberately conflicting within- and between-cell profiles.
Each covariate is a noisy linear combination of the standardized macro signal and the standardized within-cell signal, one loading predominantly on the macro signal and the other on the within-cell signal; the macro proxy additionally carries a cell-persistent noise component, so that the two covariates are not collinear across cell means.
This construction realizes the capacity trade-off characterized in Section~\ref{sec:theta}: because the within-optimal and between-optimal predictors point in different directions and a fixed, loss-independent Ridge penalty limits capacity, the training loss determines how predictive capacity is allocated across the two levels, and hence the between-cell correlation $\rho_m$ that the fitted covariate attains.\footnote{As established around Equation~\eqref{eq:perlevel}, when capacity is unconstrained, or when a single clean covariate serves each level, all reasonable losses attain the same $\rho_m$ and are therefore indistinguishable once $\theta$ is fit per level. Exhibiting a training-loss effect thus requires a genuine within-versus-between capacity trade-off, which the conflicting-covariate construction supplies.}%

We compare a two-by-two factorial of the training loss and the analysis-time coefficient.
The training loss is either unit-level MSE or the power-optimal composite loss $\Lpow$ of Equation~\eqref{eq:Lpow}, and the coefficient is either the unit-level OLS slope $\theta_{\mathrm{unit}}$ of Equation~\eqref{eq:thetaunit} or the per-level pair $(\theta_w,\theta_m)$ of Equation~\eqref{eq:perlevel}.
The four combinations range from the naive baseline (MSE with $\theta_{\mathrm{unit}}$) to the fully aligned estimator ($\Lpow$ with per-level $\theta$), with the two intermediate estimators isolating the contribution of each stage; the unadjusted difference-in-means serves as the reference.
For each configuration we perform $1{,}000$ Monte Carlo replications, assigning each regime a disjoint block of random seeds so that the regimes are statistically independent while all estimators within a regime share seeds, and report the mean standard error relative to the unadjusted estimator, the rejection rate under a fixed alternative calibrated to a naive-estimator power near $0.33$ in the balanced regime, and the false positive rate under the null; all standard errors are cluster-robust at the cluster level.%

\subsection{Results}
\label{sec:sim_results}

The alignment gains predicted by the theory materialize precisely in the regime that characterizes fine-grained switchback data.
Table~\ref{tab:sims} reports the three regimes, and Figure~\ref{fig:sims} plots the relative standard error of the naive and aligned estimators as a function of $\Smacro$.%

\begin{table}[t]
\centering
\caption{Monte Carlo comparison across macro variance shares $\Smacro$ ($1{,}000$ replications each, with a disjoint block of random seeds per regime). Standard errors are reported relative to the unadjusted estimator; power uses a fixed alternative calibrated to naive power $\approx 0.33$ in the balanced regime. The aligned estimator's advantage over the naive baseline grows as $\Smacro$ falls.}
\label{tab:sims}
\begin{tabular}{llccc}
\toprule
$\Smacro$ & Estimator & $\mathrm{SE}/\mathrm{SE}_{\mathrm{raw}}$ & Power & $\rho_m$ \\
\midrule
$0.50$ & naive (MSE, $\theta_{\mathrm{unit}}$)                 & $0.531$ & $0.33$ & $0.851$ \\
       & \quad $+$ per-level $\theta$                          & $0.531$ & $0.33$ & $0.851$ \\
       & \quad $+$ $\Lpow$                                     & $0.504$ & $0.35$ & $0.864$ \\
       & aligned ($\Lpow$, per-level $\theta$)                 & $0.502$ & $0.35$ & $0.864$ \\
\midrule
$0.25$ & naive (MSE, $\theta_{\mathrm{unit}}$)                 & $0.649$ & $0.39$ & $0.813$ \\
       & \quad $+$ per-level $\theta$                          & $0.584$ & $0.47$ & $0.813$ \\
       & \quad $+$ $\Lpow$                                     & $0.602$ & $0.43$ & $0.857$ \\
       & aligned ($\Lpow$, per-level $\theta$)                 & $0.500$ & $0.59$ & $0.857$ \\
\midrule
$0.15$ & naive (MSE, $\theta_{\mathrm{unit}}$)                 & $0.821$ & $0.41$ & $0.778$ \\
       & \quad $+$ per-level $\theta$                          & $0.620$ & $0.60$ & $0.778$ \\
       & \quad $+$ $\Lpow$                                     & $0.776$ & $0.46$ & $0.843$ \\
       & aligned ($\Lpow$, per-level $\theta$)                 & $0.507$ & $0.77$ & $0.843$ \\
\bottomrule
\end{tabular}
\end{table}

In the balanced regime $\Smacro=0.50$, in which the within- and between-cell shares are equal, the aligned estimator improves on the naive baseline only marginally, reducing the standard error by roughly five percent, and the per-level coefficient is indistinguishable from the unit-level coefficient.
This is the anticipated behavior: when the macro signal constitutes a large share of total variance, the unit-level loss already allocates substantial capacity to it, so little remains for the alignment to recover.
This regime serves as a control, confirming that the method does not manufacture gains where none are warranted.%

As $\Smacro$ decreases and within-cell noise comes to dominate, the naive estimator degrades sharply: its relative standard error rises from $0.531$ to $0.821$, approaching the unadjusted benchmark, whereas the aligned estimator remains essentially flat at approximately $0.50$ across all three regimes.
The advantage of alignment over the naive baseline therefore grows monotonically, from $5.5\%$ at $\Smacro=0.50$ to $23.0\%$ at $\Smacro=0.25$ and $38.3\%$ at $\Smacro=0.15$, with a corresponding increase in power from $0.41$ to $0.77$ in the most idiosyncratic regime.
This pattern is the empirical counterpart of the variance identity of Proposition~\ref{thm:variance}: the unit-level loss allocates capacity by total-variance share, while the estimator variance is governed by the macro share, and the two diverge exactly as $\Sres$ grows.%

The two intermediate estimators show that the training-loss and coefficient corrections are complementary and that neither alone suffices in the idiosyncratic regime.
At $\Smacro=0.15$, re-estimating the coefficient per level alone lowers the relative standard error from $0.821$ to $0.620$, and re-training under $\Lpow$ alone lowers it to $0.776$, but only their combination attains $0.507$.
The magnitude of the coefficient correction is itself increasing as $\Smacro$ falls---it is negligible at $\Smacro=0.50$ but substantial at $\Smacro=0.15$---consistent with Proposition~\ref{prop:theta_misalign}: as the within-cell share grows, the unit-level coefficient is drawn toward the within-cell slope and departs from the between-cell slope $\theta_m$, so per-level estimation becomes essential.%

The two corrections are moreover \emph{coupled}, not merely additive, and the factorial isolates the mechanism.
At $\Smacro=0.15$, replacing unit-level MSE with $\Lpow$ while retaining the unit-level coefficient lowers the relative standard error only from $0.821$ to $0.776$, whereas the same substitution under per-level estimation lowers it from $0.620$ to $0.507$: the training-loss upgrade is largely inert under $\theta_{\mathrm{unit}}$ and delivers its full benefit only once the coefficient is fit per level.
The reason is that the loss upgrade is precisely what creates the need for the coefficient upgrade.
Because $\Lpow$ reallocates capacity toward the cell-mean fit at the expense of the within-cell fit, it deliberately drives the within- and between-cell regression slopes of $g$ apart---the level-miscalibration of Corollary~\ref{cor:perlevel}---so a single coefficient, pulled toward the within-cell slope, mis-scales exactly the macro component that $\Lpow$ has improved.
Per-level estimation dissolves the coupling by scaling each level separately, at which point the calibration scale of $g$ is irrelevant and only the correlation $\rho_m$ it attains governs power.
The weighted loss and per-level coefficient are therefore best understood as a single aligned procedure: the loss builds macro predictive quality, and per-level $\theta$ is what converts that quality into realized variance reduction.%

The between-cell correlation $\rho_m$ corroborates the mechanism of Section~\ref{sec:theta}.
The power-optimal loss raises $\rho_m$ in every regime---for example from $0.778$ to $0.843$ at $\Smacro=0.15$---and the reduction in standard error tracks this increase, as the theory requires given that the attainable variance depends on the covariate only through $\rho_w$ and $\rho_m$.
Across all configurations the false positive rate remains near the nominal five percent level, ranging from $0.05$ to $0.07$, and is essentially unchanged between the naive and aligned estimators, confirming that the alignment reduces variance without compromising test validity.%

\begin{figure}[t]
\centering
\includegraphics[width=0.72\textwidth]{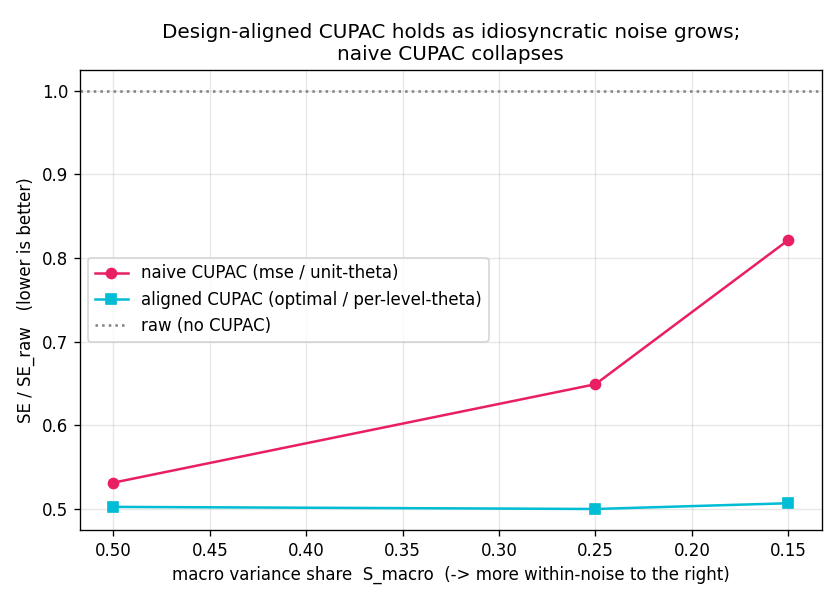}
\caption{Relative standard error $\mathrm{SE}/\mathrm{SE}_{\mathrm{raw}}$ as a function of the macro variance share $\Smacro$ (within-cell noise increases to the right). The naive estimator (unit-level MSE with $\theta_{\mathrm{unit}}$) degrades toward the unadjusted benchmark as $\Smacro$ falls, whereas the aligned estimator ($\Lpow$ with per-level $\theta$) remains stable; the single-stage corrections fall in between. The advantage of alignment is concentrated in the idiosyncratic-noise-dominated regime characteristic of fine-grained switchback data.}
\label{fig:sims}
\end{figure}

\section{Discussion}
\label{sec:discussion}

The result is intuitive in hindsight: experiment power in a switchback is set by the between-cell residual variance, so the control variate should be trained to predict \emph{cell means}, not individual events---and the precise exchange rate between the two goals is $1+\nbar(1+\cv^2)$.
Empirically we expect the largest gains for metrics with (i) large within-cell idiosyncratic noise (low intraclass correlation) and (ii) high cell density $\nbar$, where unit-level MSE is most distracted by within-cell fluctuations.
A natural validation is to train, for primary operational metrics (such as service duration or volume) and per segment, both plain MSE and $\Lpow$ control variates and compare their realized between-cell correlation $\rho_m$ and projected confidence-interval reduction on a held-out period.
The controlled simulation study of Section~\ref{sec:sims} provides such a validation on synthetic data; we leave the corresponding study on production metrics, and the integration of $\Lpow$ into the production training workflow, to future work.

The composite loss confers two distinct benefits, according to whether model capacity binds.
Its primary achievement is to optimize power: when capacity or regularization forces a trade-off
between within-cell and between-cell fit---the regime of the simulation study in
Section~\ref{sec:sims}---$\Lpow$ raises the between-cell correlation $\rho_m$ that a unit-level loss
would sacrifice, and thereby lowers estimator variance relative to any objective that fixes the
within-versus-between weighting at $1{:}1$.
Its secondary achievement is to train the control variate efficiently: when capacity is ample---for
instance, a gradient-boosted model trained without early stopping or depth limits---the covariate can
in principle drive both loss components toward their floor, so that all reasonable losses converge to
the same per-level correlations $\rho_w,\rho_m$ and hence to the same attainable variance reduction,
as the per-level analysis of Equation~\eqref{eq:perlevel} implies.
Even then $\Lpow$ retains value: by up-weighting the macro component of the training objective, it
directs a fixed training budget toward the
between-cell signal that governs power, so that a target between-cell correlation is attained in fewer
iterations and with less capacity expended on within-cell fluctuations that the design averages away.
A systematic study of this training-efficiency benefit, as distinct from the power-optimization
benefit demonstrated here, is left to future work.%

A natural question is whether the macro term should be decomposed further, replacing the
two-component loss---macro versus residual---with four components targeting the cluster, time-of-day,
interaction, and residual variances separately.
Under the variance functional of Proposition~\ref{thm:variance} this refinement yields no additional
benefit: the three macro sub-components share the identical multiplier $\tfrac{1}{\nbar}+1+\cv^2$, so
grouping them into a single macro term is already variance-optimal, and only the macro-versus-residual
split affects power.
This grouping is robust to temporal structure, because cell-level randomization is independent across
time windows and therefore neutralizes temporal autocorrelation in the outcome, leaving the macro
multiplier unchanged~\citep{pankratev2026powerful}.%

\section{Conclusion}
\label{sec:conclusion}

In this paper, we have proposed a switchback-aligned CUPAC framework that adapts the CUPAC workflow---covariate training and experiment analysis---to cluster$\times$time-window designs.
While standard control variates are trained on unit-level MSE and adjusted with the unit-level OLS coefficient, switchback power is dictated by between-cell variance; we proved that each default stage under-weights this macro component by $1+\nbar(1+\cv^2)$, typically two to three orders of magnitude.
This methodology replaces both default stages with objectives tied to the same variance functional: the power-optimal composite loss $\Lpow$, and a design-weighted adjustment coefficient estimated under the same loss.
We showed that $\Lpow$ is a tight upper bound on deployed variance, gave gradient--Hessian and closed-form Ridge formulations, and analyzed when per-level $\theta$ is required.
Together, these components turn ad-hoc feature engineering into a coherent, mathematically grounded workflow that directly optimizes experiment power.

\bibliographystyle{plainnat}
\bibliography{references}

\vspace{1cm}
\appendix
\noindent{\huge \textbf{Appendix}}

\section{Notation Summary}
\label{app:notation}
\begin{center}
\begin{tabular}{ll}
\toprule
Symbol & Meaning \\
\midrule
$J,\,H,\,B=JH$ & clusters, time windows, cells \\
$b=(\mathrm{cl},t),\ i$ & cell (cluster $\times$ time) index and within-cell unit index ($i=1,\dots,n_b$) \\
$n_b,\ \nbar,\ \cv^2$ & cell size, mean cell density, squared CV of cell sizes \\
$\sigtot$ & total outcome variance \\
$\Scl,\Stime,\Sint,\Sres$ & variance shares (cluster, time, interaction, residual) \\
$\Smacro=\Scl+\Stime+\Sint$ & cell-level (macro) share, $\Sres=1-\Smacro$ \\
$g_{b,i}$ & control-variate prediction for unit $i$ in cell $b$ \\
$e_{b,i}=y_{b,i}-g_{b,i},\ \ebar_b$ & residual and its cell mean \\
$\MSEw,\ \MSEm$ & within-cell and between-cell mean squared error \\
$\theta,\ \rho_k$ & CUPED coefficient; level-$k$ outcome--prediction correlation \\
$\lambda=\nbar(1+\cv^2)$ & macro-penalty weight \\
\bottomrule
\end{tabular}
\end{center}

\section{Proofs}
\label{app:proofs}

This appendix collects the proofs of the results stated in Section~\ref{sec:theory}.

\begin{proof}[Proof of Proposition~\ref{thm:variance}]
Under Assumption~\ref{ass:calib} the adjusted outcome is the residual $e_{b,i}=y_{b,i}-g_{b,i}$, which inherits the multilevel structure~\eqref{eq:model} with each component replaced by the corresponding component of $e$.
Substituting the residual variance components into the raw power formula~\eqref{eq:rawvar} and identifying $\sigtot\Sres \mapsto \sigma^2_{\mathrm{res},e}=\E[(e_{b,i}-\ebar_b)^2]=\MSEw$ (the within-cell variance of the residual) and $\sigtot\Smacro \mapsto \sigma^2_{\mathrm{macro},e}=\Var(\ebar_b)=\E[\ebar_b^2]=\MSEm$ (the between-cell variance of the residual, centered by calibration) yields~\eqref{eq:cupacvar}.
\end{proof}

\begin{proof}[Proof of Corollary~\ref{prop:misalign}]
For the unweighted unit-level loss, $\MSE_{\mathrm{total}}=\MSEw+\MSEm$ by~\eqref{eq:totaldecomp}, so it assigns the within- and between-cell losses equal effective weight, $a=b=1$.
The design-required ratio~\eqref{eq:ratio} then follows directly from the weights~\eqref{eq:abweights}.
\end{proof}

\begin{proof}[Derivation of the matched coefficient~\eqref{eq:thetastar}]
Collect the within- and between-cell second moments of $(y,g)$,
\begin{equation}
\label{eq:moments}
\begin{aligned}
V^w_y &= \E[(y_{b,i}-\ybar_b)^2], & V^w_g &= \E[(g_{b,i}-\gbar_b)^2], & C^w &= \E[(y_{b,i}-\ybar_b)(g_{b,i}-\gbar_b)],\\
V^m_y &= \E[\ybar_b^{\,2}], & V^m_g &= \E[\gbar_b^{\,2}], & C^m &= \E[\ybar_b\,\gbar_b],
\end{aligned}
\end{equation}
with cell means centered to mean zero across cells.
Since $r_{b,i}(\theta)=y_{b,i}-\theta g_{b,i}$ is linear in $\theta$, its within- and between-cell losses are $\MSEw(r)=V^w_y-2\theta C^w+\theta^2 V^w_g$ and $\MSEm(r)=V^m_y-2\theta C^m+\theta^2 V^m_g$, so by Proposition~\ref{thm:variance}
\begin{equation}
\label{eq:Vtheta}
\mathcal{V}(\theta) := \frac{JH}{4}\,\Var\!\big(\tauhat_{\mathrm{CUPAC}}(\theta)\big)
= a\big(V^w_y-2\theta C^w+\theta^2 V^w_g\big) + b\big(V^m_y-2\theta C^m+\theta^2 V^m_g\big).
\end{equation}
This is exactly the composite training loss $\Lpow(y-\theta g)$; it is strictly convex in $\theta$, and $\mathcal V'(\theta)=0$ gives~\eqref{eq:thetastar}.
The unit-level slope~\eqref{eq:thetaunit} is the special case $a=b$, and minimizing the two terms of~\eqref{eq:Vtheta} separately gives the per-level optimum~\eqref{eq:perlevel}.
\end{proof}

\begin{proof}[Proof of Theorem~\ref{thm:upperbound}]
At level $k$, the optimal-$\theta$ residual variance is $\sigma^2_{y,k}(1-\rho_k^2)$, whereas the MSE component is $\sigma^2_{y,k}\big(1-2\rho_k s_k + s_k^2\big)$ with $s_k=\sigma_{g,k}/\sigma_{y,k}$.
Since $1-2\rho_k s_k+s_k^2 \ge 1-\rho_k^2$ for all $s_k$, with equality at $s_k=\rho_k$, each MSE component dominates the optimal-$\theta$ residual variance.
Summing the level contributions with the design weights of~\eqref{eq:cupacvar} gives the claim.
\end{proof}

\end{document}